\documentclass[12pt,letterpaper]{article}
\usepackage[top=0.85in,left=1in,footskip=0.5in,marginparwidth=2in]{geometry}
\usepackage{color}

\usepackage{comment}
\usepackage{amsmath,amssymb}
\usepackage{multirow}
\usepackage{url}

\newtheorem{problem}{Problem}
\usepackage{newtxtext,newtxmath}
\usepackage[T1]{fontenc}

\usepackage[numbers,sort]{natbib}

\usepackage{setspace}

\usepackage{nameref,hyperref}

\usepackage[right]{lineno}

\usepackage{microtype}
\DisableLigatures[f]{encoding = *, family = * }

\raggedright
\usepackage{changepage}

\usepackage[aboveskip=1pt,labelfont=bf,labelsep=period,singlelinecheck=off]{caption}

\makeatletter
\renewcommand{\@biblabel}[1]{\quad#1.}
\makeatother

\usepackage{lastpage,fancyhdr,graphicx}
\usepackage{epstopdf}
\renewcommand{\footrule}{\hrule height 2pt \vspace{2mm}}
\fancyheadoffset[L]{2.25in}
\fancyfootoffset[L]{2.25in}

\usepackage{color}

\definecolor{Gray}{gray}{.25}

\usepackage{graphicx}

\usepackage{sidecap}
\usepackage{ragged2e}

\usepackage{wrapfig}
\usepackage[pscoord]{eso-pic}
\usepackage[fulladjust]{marginnote}
\reversemarginpar

\usepackage{algorithm}
\usepackage{algpseudocode}
\newtheorem{theorem}{Theorem}
\newtheorem{lemma}{Lemma}

\newenvironment{proof}{
    \vspace{-1.5em}
    \paragraph{\normalfont\itshape Proof.}
    \normalfont
}{
    \hfill$\square$
    \vspace{1em}
}

\hypersetup{
hidelinks
}
\newcommand{\keywords}[1]{
\begin{flushleft}
\textbf{Keywords:} #1
\end{flushleft}
}

\begin{document}
\pagestyle{plain}
\pagestyle{myheadings}
\fancyhf{}
\rfoot{\thepage/\pageref{LastPage}}
\renewcommand{\footrule}{\hrule height 2pt \vspace{2mm}}
\fancyheadoffset[L]{2.25in}
\fancyfootoffset[L]{2.25in}
\vspace*{0.35in}

\begin{center}
{\Large
\textbf\newline{RatGene: Gene deletion-addition algorithms using growth to production ratio for growth-coupled production in constraint-based metabolic networks} 
}
\newline
\\
Yier Ma$^{1,2}$, Takeyuki Tamura$^{1,2}$ \\
\end{center}

\begin{flushleft}
$^1$Bioinformatics Center, Institute for Chemical Research,
Kyoto University, Kyoto, Japan \\
$^2$Graduate School of Informatics, Kyoto University, Kyoto, Japan \\
\texttt{Email: mayier@kuicr.kyoto-u.ac.jp, tamura@kuicr.kyoto-u.ac.jp}
\end{flushleft}
\bigskip

\justifying
\section*{Abstract}
In computational metabolic design, it is often necessary to modify the original constraint-based metabolic networks to lead to growth-coupled production, where cell growth forces target metabolite production. 
However, in genome-scale models, finding strategies to simultaneously delete and add genes to induce growth-coupled production is challenging. 
This is particularly true when heavy computation is necessary due to numerous gene deletions and additions.
In this study, we mathematically defined related problems, proved NP-hardness and/or NP-completeness, and developed an algorithm named RatGene that (1) automatically integrates multiple constraint-based metabolic networks, (2) identifies gene deletion-addition strategies by a growth-to-production ratio-based approach, and
(3) eliminates redundant gene additions and deletions.
The results of computational experiments demonstrated that the RatGene-based approach can significantly improve the success ratio for identifying the strategies for growth-coupled production.
RatGene can facilitate a more rational approach to computational metabolic design for the production of useful substances using microorganisms by concurrently considering both gene deletions and additions.

\keywords{Biochemistry, Integer linear programming, Constraint optimization}

\section{Introduction}
Computational metabolic design requires adjusting the metabolic functions of microorganisms to create efficient biochemical pathways to produce beneficial substances. 
By adding or deleting genes, modifying enzymes, or adjusting cellular processes, the production of food, biofuels, large-scale chemicals, pharmaceuticals, and other bioproducts can be enhanced \cite{lee2015systems,nielsen2016engineering,peralta2012microbial}. 
Understanding the core principles of metabolism and cell activities is critical in this field, and heavily relies on mathematical modeling techniques. These techniques involve creating mathematical models that depict the behavior of metabolic systems. Such models allow us to simulate and predict how these systems respond under various conditions, such as the creation of a new metabolic pathway or shifts in nutrient availability.

{\bf Constraint-based models} of metabolic networks are often used in the design of genome-scale metabolic networks for growth-coupled production. 
This model assumes a balanced metabolism characterized by stable flux distributions and consistent metabolite concentrations. 
This implies that the total production rates are equal to the total consumption rates for each metabolite.
This assumption simplifies the modeling and analysis of metabolic systems, promoting the use of mathematical models to design genome-scale metabolic networks for growth-coupled production.
In constraint-based models, reactions are classified into two categories: internal and external reactions. 
Internal reactions simultaneously produce and consume metabolites, while external reactions either consume or produce metabolites, but not both at the same time.
For each reaction rate, lower and upper bounds are provided, with negative values being allowed for reversible reactions.

The {\bf cell growth reaction} and the {\bf target metabolite production reaction} are of particular interest in the constraint-based models. 
The cell growth reaction is a virtual reaction designed to match the cell growth rate of the biological experiments. 
It reflects the efficient conversion of absorbed resources into cellular energy and chemical constituents, and the facilitation of cellular evolution under selection pressure \cite{maranas2016optimization}. 
The constraint-based models prioritize maximizing cell growth for this reason. The reaction generating the target metabolite is called the target metabolite production reaction.
Such analysis based on the constraint-based models is called flux balance analysis (FBA), which limits biomass production and nutrient availability, as well as focuses on target chemical synthesis. The aim of the FBA-based methods is to investigate and optimize metabolic flux distributions \cite{orth2010flux}.

A common computational task in metabolic engineering is to design constraint-based models for {\bf growth-coupled production}. 
Let {\bf GR} and {\bf PR} represent the rates of cell growth and target metabolite reactions, respectively.
In this study, growth-coupled production is defined as a condition where the minimum values of PR and GR are 0.001 mmol/gDW/h or more when GR is maximized.
Metabolic network design is often identified by reaction deletion strategies from the original constraint-based models. 

One of the most efficient methods to determine the reaction deletion strategies for growth-coupled production is the elementary mode (EM) and/or minimum cut set (MCS)-based methods, which utilize non-decomposable steady-state fluxes in the design of constraint-based models \cite{ballerstein2011minimal,schuster1994elementary}. 
The MCS-based method developed by von Kamp et al. demonstrated that growth-coupled production was feasible for almost all metabolites under appropriate conditions in genome-scale metabolic models of five key industrial species \cite{von2017growth}.

However, reaction deletions are realized by gene deletions because many chemical reactions are catalyzed by enzymes encoded by genes in metabolic networks. 
Identifying gene deletion strategies for growth-coupled production remains challenging due to intensive computation, especially with complex {\bf gene-protein-reaction (GPR) rules} and many necessary gene deletions in genome-scale models \cite{machado2016stoichiometric,razaghi2020genereg,tamura2022trimming}.
In the constraint-based models, GPR rules represent the relationships between genes and reactions with Boolean functions, where the inputs are genes and the outputs are reactions.
When a reaction is inhibited due to gene deletions via GPR rules, both the lower and upper bounds of the reaction rate are forced to be zero.

Recently, Tamura et al. developed gDel\_minRN by considering GPR rules to identify gene deletion strategies to extract the core part for growth-coupled production \cite{tamura2023gene}.
However, the success ratio of gDel\_minRN still has room for improvement.
Introducing genes from other species is another powerful approach to improve the success ratio.
If we can combine multiple constraint-based models and obtain appropriately a larger reference network that contains all potential additions, the modification strategies consisting of gene additions and deletions can be identified through gene deletions from the larger network.

In this study, we developed RatGene, which (1) integrates two constraint-based models $N_1$ and $N_2$, (2) identifies modification strategies, consisting of gene additions and deletions from $N_1$, for growth-coupled production, and (3) reduces the size of the modification strategies.
In the computational experiments, the performance of RatGene on various datasets was compared with that of existing methods, GDLS and gDel$\_$minRN \cite{lun2009large,tamura2023gene}.
The results of the computational experiments showed that (1) for many target metabolites, RatGene can calculate a gene deletion-addition strategy even when gDel\_minRN cannot, 
(2) by using gDel\_minRN and RatGene in a complementary manner, gene deletion-addition strategies that lead to growth-coupled production can be calculated for more target metabolites, and (3) RatGene is efficient for identifying gene deletion strategies (without gene addition) as well. 

The remaining of this paper is as follows:
Section~\ref{bg} briefly summarizes the constraint-based models;
Section~\ref{subsec:prodef} defines problems and proves NP-hard and NP-complete; 
Section 2.3 explains the workflow of the developed algorithm RatGene;
Section~\ref{subsec:integrate} describes the integration process of two models;
Section~\ref{subsec:strgen} introduces how RatGene determines the initial modification strategies;
Section~\ref{subsec:sizred} states the process to minimize the size of the modification strategies.
Then, computational experiments conducted in this study are reported in Section~\ref{sec:comexp}:
Sections 3.1 and 3.2 are for gene deletion-addition and deletion problems, respectively.
Sections \ref{subsec:discuss}, \ref{subsec:conclu}, and \ref{subsec:review} are for discussion, conclusion, and related works, respectively.

\section{Method}\label{sec:method}
\subsection{Constraint-based Models of Metabolic Networks}\label{bg}
A constraint-based model $N = \{ R, M, G, S, lb, ub, h \}$ of metabolic networks is composed of a set of \textbf{reactions} $R$, a set of \textbf{metabolites} $M$, a set of \textbf{genes} $G = \{g_1,\ldots,g_l\}$, a \textbf{stoichiometric matrix} $S$, a set of \textbf{lower bounds} of reaction rates $lb$, a set of \textbf{upper bounds} of reaction rates $ub$, and a set of \textbf{GPR} rules $h$. 
In the context of constraint-based models of metabolic networks, the important thermodynamic properties of the system are encoded within the stoichiometric matrix $S$. 
This matrix is of dimensions $m \times n$, where $m$ and $n$ represent the number of metabolites and reactions in the network, respectively. 
A \textbf{mode} is a flux vector $v \in \mathbb{R}^n$ in the null space of $S$, $N(S)=\{v|S \cdot v=0\}$. This implies that the consumption and production rates of any metabolite in the network are equal in order to maintain the stability of its concentrations:
\begin{equation}\label{cs}
    S \cdot v=0
\end{equation}
For accurate representation of the reaction rates of $n$ reactions, the following inequalities must hold true:
\begin{equation}
    lb_i \leq v_i \leq ub_i \quad i \in \{1, 2, \ldots, n\}
\end{equation}
This imposes limitations on the lower bound and the upper bound of the reaction rates, respectively, as the magnitude of the reaction rate in the system cannot be infinitely high or low. 
Furthermore, GPR rules given by Boolean functions are represented as follows:
\begin{equation}\label{e:6}
    p_i=h_i(G)
\end{equation}
Each element of $G$ is a binary variable, taking on a value of either 0 or 1.
$h$ is a Boolean function that typically has two forms:
\begin{align}
    & h_i = \bigwedge_{j=1}^{\lambda_j} \gamma_j, \quad \gamma_j \in \{ g, h \}  \\
    & h_i = \bigvee_{j=1}^{\lambda_j} \gamma_j, \quad \gamma_j \in \{ g, h \}
\end{align}
Here, $\gamma$ is either a gene $g$ or another Boolean function $h$, and $\lambda$ is the number of entities in $h$. 
These two forms can be represented by the following two linear inequalities:
\begin{align}
    & 1 - \lambda + \sum \gamma \leq \lambda \cdot h_i \leq \sum \gamma  \\
    & \sum \gamma \leq \lambda \cdot h_i \leq \lambda \cdot \sum \gamma 
\end{align}
The inequalities $(2)$ then can be modified as:
\begin{equation}\label{ce}
    p_i \cdot lb_i \leq v_i \leq p_i \cdot ub_i,
\end{equation}
where $p_i$ indicates whether the reaction $v_i$ is deleted.
In FBA, the vector space \( V \) is defined by inequalities \( (\ref{cs}) \) to \( (\ref{ce}) \). This vector space represents the feasible solution space. Within this space, the objective is to find an optimal solution that satisfies the objective function of the mixed integer linear programming (MILP) problem \( P \):

\begin{align}
    &max \quad f(v) \\
    &s.t. \nonumber \\
    &\quad S \cdot v=0 \nonumber \\
    &\quad p_i \cdot lb_i \leq v_i \leq p_i \cdot ub_i \nonumber \\
    &\quad p_i=h_i(g), \nonumber
\end{align}
where $f(v)$ is the objective function.
The problem for maximizing $f(v)$ on the vector space $V$ is denoted by $P_V^{f(v)}$.
When $f(v)$ is GR, the problem is denoted as $P_V^{v_{growth}}$, where the cell growth reaction rate $v_{growth}$ is maximized.
The solution $x_V^{f(v)}$ represents a vector that indicates the rates of $n$ reactions when maximizing $f(v)$ in $V$. 
We need to identify a subspace $U \subseteq V$ such that the solution $x_U^{f(v)}$ of the problem $P_U^{f(v)}$ satisfies GR$\geq$0.001 and PR$\geq$0.001. 
The subspace $U$ represents a space induced by a modification strategy, specifically a deletion strategy in this case.
The modification strategy $D$ is a 0/1 assignment for each gene, which represents either deletion or retention of the associated gene. 
Any space derived from a specific \( D \) is a subspace of the original vector space \( V \).

For example, consider a constraint-based model where \( l = 3 \), \( n = 2 \), and the GPR rules are given by \( p_1 = g_1 \land g_2 \) and \( p_2 = g_1 \lor g_3 \).
Let us consider a modification strategy \( D_e = (1,0,0) \), which implies \( g_1 = 1 \), \( g_2 = 0 \), and \( g_3 = 0 \).
The following illustrate the original feasible solution space \( V_e \) and the subspace \( U_e \) derived from the modification strategy \( D_e \), where genes \( g_2 \) and \( g_3 \) are deleted.

\begin{minipage}{0.45\linewidth}
\begin{align*}
    &\mathbf{V_e} \\
    &S \cdot v=0 \\
    &p_1 \cdot lb_1 \leq v_1 \leq p_1 \cdot ub_1 \\
    &p_2 \cdot lb_2 \leq v_2 \leq p_2 \cdot ub_2 \\
    &p_1=g_1 \wedge g_2 \\
    &p_2=g_1 \vee g_3 \\
\end{align*}
\end{minipage}
\hfill
\begin{minipage}{0.45\linewidth}
\begin{align*}
    &\mathbf{U_e} \\
    &S \cdot v=0 \\
    &v_1=0 \\
    &lb_2 \leq v_2 \leq ub_2 \\
    &p_1=0 \\
    &p_2=1 \\
\end{align*}
\end{minipage}

Figure \ref{fig:net1}(A) provides a toy example network to describe the assumption of the flux balance and the growth-coupled production. This simple example network contains seven reactions \{R1, R2, R3, R4, R5, R6, R7\} and three metabolites \{C1, C2, C3\}, where R1 is the input to the network and R6 and R7 are the outputs of the network. R6 and R7 are the cell growth and the target metabolite production reactions, respectively. The upper and lower bounds of all reactions are attached as $[\epsilon_x,\epsilon_y]$ in the figure. 
Figure \ref{fig:net1}(B) illustrates the stoichiometric matrix $S$ corresponding to this example network.
Each row represents the relationship between a metabolite and all reactions, and each entry indicates the coefficient by which that metabolite is produced or consumed per occurrence of that reaction.
For example, the first row indicates that C1 is generated by R1 but consumed by R2, R3, and R4. 
The assumption of flux balance means that the concentration of any metabolite is stable. 
For instance, the inner product of the first row of $S$ and a flux vector $[2,1,1,0,0,0,0]^{T}$ ensures the change of the concentration of C1 is 0. 
One of the valid flux vectors for the wild type is given by \( [2,2,0,0,0,2,0]^{T} \). Since the flux of R2 can be substituted with R3 and R5, for example, the vector \( [2,0,2,0,2,2,0]^{T} \) is also valid.

Figure \ref{fig:net1}(C) is a table representing deletion strategies and the resulting flux rates.
The deletion strategy \{R3,R5\} induces the flux vector $[2,0,0,2,0,2,2]^{T}$ in the optimistic scenario (best-case) regarding R7.
This is the state of growth-coupled production since
the target metabolite production reaction rate is not zero when the growth reaction is maximized under the assumption of flux balance. 
However, this deletion strategy would not ensure a constant growth-coupled production because the flux vector $[2,2,0,0,0,2,0]^{T}$ is obtained in the pessimistic scenario (worst-case) regarding R7.  
To ensure growth-coupled production even in the worst-case scenario for the target metabolite production, we should delete \{R2, R5\}.
And this deletion strategy induces the flux vector $[2,0,0,2,0,2,2]^{T}$, which ensures the growth-coupled production even in the worst case.
\begin{figure}[ht]
    \centering
    \includegraphics[scale=0.5]{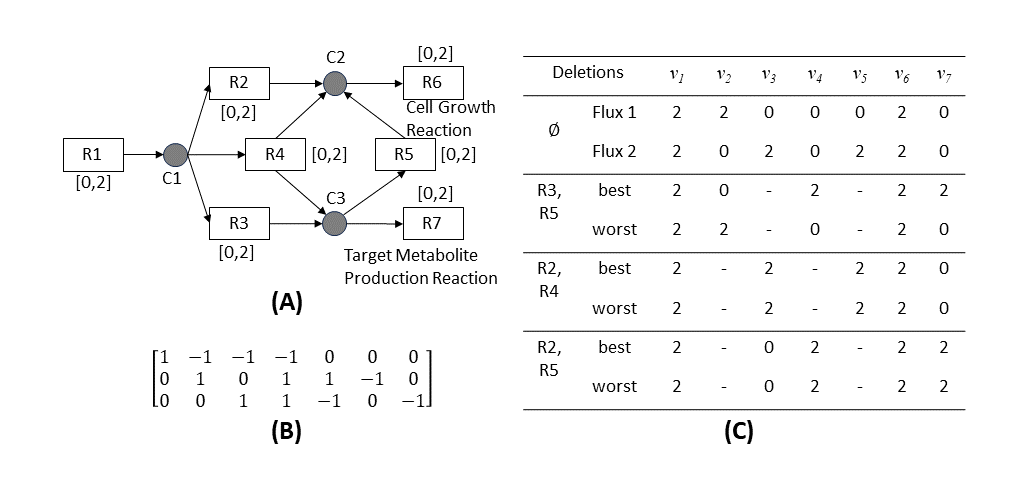}
    \captionsetup{justification=centering}
    \caption{(A) An example network. (B) The stoichiometric matrix. (C) A table represents deletion strategies and flux rates.}
    \label{fig:net1}
\end{figure}

\subsection{Problem Definition}\label{subsec:prodef}

We consider four types of problems for finding \( x_U^{f(v)} \) in \( P_U^{f(v)} \). 
The gene deletion problem for maximizing the minimum $v_{target}$ is defined as follows, where $v_{target}$ is the target metabolite production reaction.

\begin{problem}\label{p1}
Find $D$ that derives the subspace $U$ of the original vector space $V$, subject to $U=\underset{U \subseteq V} {\operatorname{argmax}} \ (\operatorname{min}v_{target})$ in $P^{v_{growth}}_{V}$.
\end{problem}
Instead of maximizing the minimum PR with GR maximization, its decision problem version is defined with the thresholds $lb^{threshold}_{target}$ and $lb^{threshold}_{growth}$ for $v_{target}$ and $v_{growth}$.

\begin{problem}\label{p2}
({\bf Prob-gDel}): 
Find $D$ that derives the subspace $U$ of the original vector space $V$, subject to $v_{target} \geq lb^{threshold}_{target}$ and $v_{growth} \geq lb^{threshold}_{growth}$ for any $\ v_{target}$ in $P^{v_{growth}}_{V}$. "No solution" is returned if there is no such $D$.
\end{problem}
When $lb^{threshold}_{target}$ gradually increases, Problem \ref{p2} approaches to Problem \ref{p1}. 

\begin{theorem}\label{th1}
Problem \ref{p1} is NP-hard.
\end{theorem}
\begin{proof}
The NP-hardness is proved by the reduction from the 3-SAT problem \cite{sipser1996introduction}. 
Any instance of the 3-SAT problem can be represented by the constraint-based models as shown in Figure \ref{fig:proof}. 
The rectangles and circles with solid borders are reactions and metabolites, respectively. 
All coefficients related to reactions are set to 1 unless otherwise specified.
Figure \ref{fig:proof}(A) shows the representation of a literal. 
$g_i$ is a gene that controls the reaction $x_k$: $g_i=1$ makes $x_k$ active while $g_i=0$ makes $x_k$ inactive.
Suppose that $g_i=1$ holds.
Since $\bar{x_i}=x_k=1$ always holds,
all flow from $\bar{x_i}$ reaches $x_k$, but not to $\bar{x_j}$.
All flow from $x_i$ reaches $x_j$, but not to $\bar{x_j}$.
Then, $x_j=1$ and $\bar{x_j}=0$ hold.
Next, Suppose that $g_i=0$ holds.
Because $x_k=0$ and $\bar{x_i}=1$ must hold,
all flow from $\bar{x_i}$ must be consumed by $\bar{x_j}$.
The flow from $x_i$ is consumed by $\bar{x_j}$ and cannot reach $x_j$.
Then, $x_j=0$ and $\bar{x_j}=1$ hold.
Thus, The network of Figure \ref{fig:proof}(A) ensures that $l_i$ is produced from $x_j$ when $g_i=1$ and from $\bar{x_j}$ when $g_i=0$. 

Figure \ref{fig:proof}(B) shows the representation of a clause of the 3-SAT instance. The rectangles with dashed borders are networks of literals shown in (A). 
The metabolite $u$ is produced by $x_j$ for a positive literal of the original 3-SAT while by $\Bar{x_j}$ for a negative literal in Figure \ref{fig:proof}(A).
$C_m$ represents a reaction $u \rightarrow y$, $C_p$ represents a reaction $2\cdot u \rightarrow y$, and $C_n$ represents a reaction $3\cdot u \rightarrow y$.
This ensures the existence of the case that one $y$ is produced when at least one literal is true.
Specifically, $C_m$, $C_p$, or $C_n$ is used depending on whether the number of satisfied literals is 1, 2, or 3, respectively.

Using the sub-networks illustrated in  Figure \ref{fig:proof}(A) and (B), any 3-SAT problem can be represented by the constraint-based model of  Figure \ref{fig:proof}(C). Rectangles with dashed borders are clauses formed by (B). $PR$ and $GR$ are the target metabolite production and the cell growth reactions, respectively.
Growth-coupled production, defined by $x_{growth} \geq 0.001$ and $x_{target} \geq 0.001$ hold for any $x_{target} \in x_{U}^{v_{growth}}$, is achieved only when $GC=GR=PR=1$ holds.
Since $l_i$, $l_j$, and $l_k$ are always 1, growth-coupled production is achieved if and only if every $y$ is 1.
Therefore, finding a solution for the 3-SAT problem is equal to finding a gene deletion strategy in the constraint-based model.
Because this constraint-based model can be obtained from any 3-SAT problem in polynomial time, the original problem is NP-hard.
\end{proof}

\begin{figure}
    \centering
    \includegraphics[scale=1]{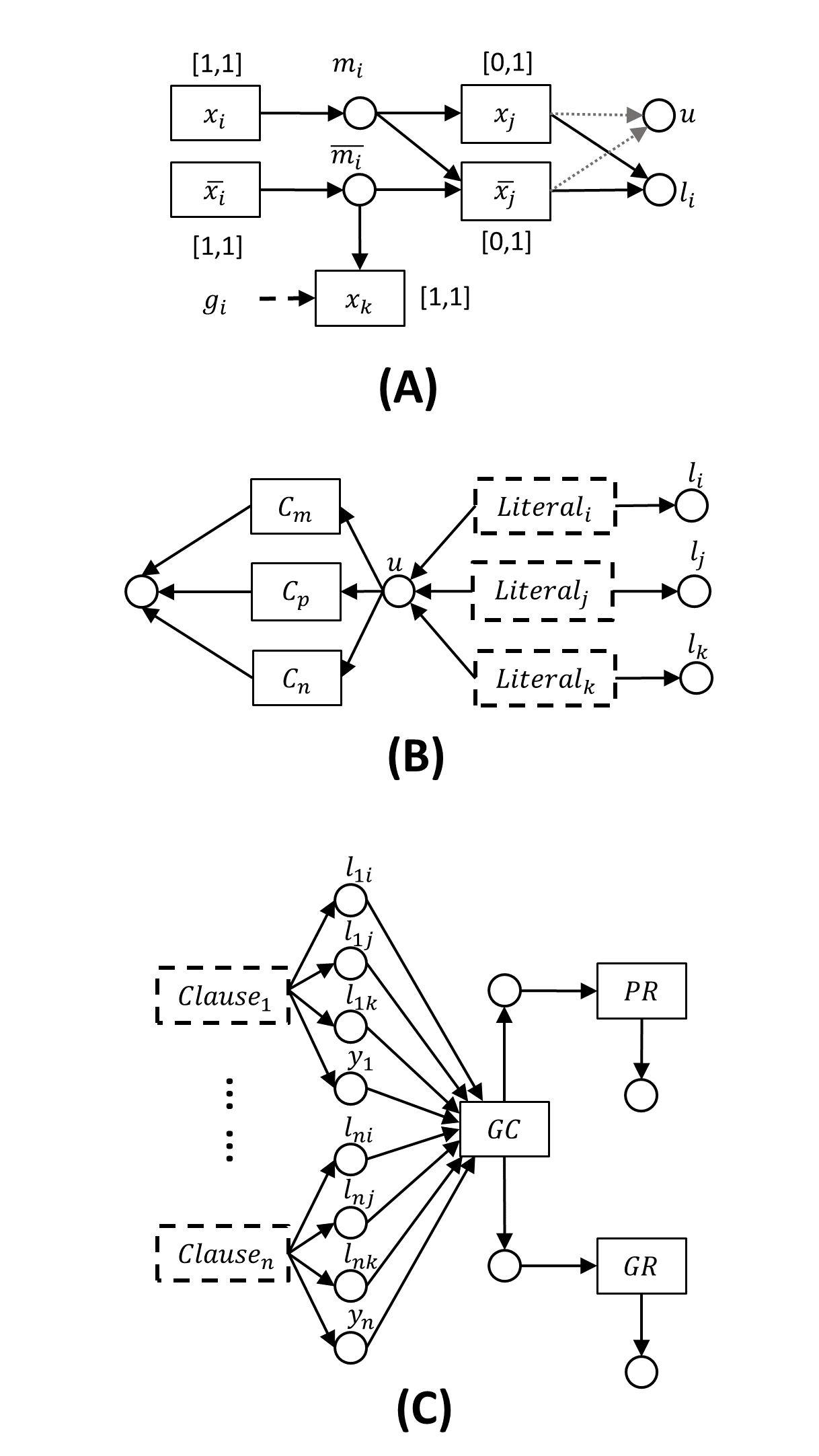}
    \captionsetup{justification=centering}
    \caption{How to convert a 3-SAT problem into the gene deletion problem in a constraint-based model. (A) A network representing a literal. (B) A network representing a clause. (C) A network representing a 3-SAT problem.}
    \label{fig:proof}
\end{figure}

The only distinction between Problems \ref{p1} and \ref{p2} is that the former is an optimization-based problem while the latter is a decision-based problem. Consequently, the following lemma can be derived.

\begin{lemma}\label{th2}
Problem \ref{p2} is NP-complete.
\end{lemma}
\begin{proof}
Problem \ref{p2} is NP-hard even when ${lb^{threshold}_{growth}}$ and ${lb^{threshold}_{target}}$ are limited to 1 according to the proof in Theorem \ref{th1}.
Because any solution can be verified in polynomial time, Problem \ref{p2} is in NP.
Since Problem \ref{p2} is in NP and NP-hard, the NP-completeness is proved.
\end{proof}

\noindent
Problems \ref{p1} and \ref{p2} focus on deletion strategies. However, in the following, we extend the problem definitions to simultaneously consider both deletion-addition strategies.

We consider a larger constrain-based model
$N_I = \{ R_I, M_I, G_I, S_I, lb_I, ub_I, h_I \}$ in addition to the original
$N = \{ R, M, G, S, lb, ub, h \}$, where $R \subseteq R_I$, $M \subseteq M_I$, and $G \subseteq G_I$ hold.
Let $V_I$ and $V$ be the vector spaces for $N_I$ and $N$, respectively.
In the extended problem,
a deletion strategy $D_I$ induces a subspace $U_I$ of $V_I$.
$U_I$ consists of a subspace $U$ of $V$ and an additional space $U^{'}$, which is a subspace of $N_I - N$.
This additional space $U^{'}$ consists of additional $n^{'}$ reactions, $m^{'}$ metabolites, and $l^{'}$ genes.
The union $U \cup U^{'}$ needs to simultaneously satisfy the constraints $(\ref{cs})$ to $(\ref{ce})$.

Then, Problem \ref{p1} is converted to as follows.
\begin{problem}\label{p3}
Find $D_I$ that derives the subspace $U_I = U \cup U^{'}$ of $V_I$, subject to $U_I=\underset{U \subseteq V, U^{'} \nsubseteq V} {\operatorname{argmax}} \ (\operatorname{min}v_{target})$ in $P^{v_{growth}}_{V_I}$.
\end{problem}
Similarly, we can define the decision version of Problem \ref{p3}.
\begin{problem}\label{p4}
({\bf Prob-gDel-Add}): 
Find $D_I$ that derives the subspace $U_I = U \cup U^{'}$ of $V_I$, subject to $v_{growth} \geq lb^{threshold}_{growth}$ and $v_{target} \geq lb^{threshold}_{target}$ for any
$v_{target}$ in $P^{v_{growth}}_{V_I}$. "No solution" is returned if there is no such $D_I$.
\end{problem}
The following lemma is held for Problems \ref{p3} and \ref{p4}.
\begin{lemma}\label{th3}
Problem \ref{p3} is NP-hard, and Problem \ref{p4} is NP-complete.
\end{lemma}
\begin{proof}
When $N_I=N$, Problems \ref{p3} and \ref{p4} are equivalent to Problems \ref{p1} and \ref{p2}, respectively. Therefore, Problems \ref{p3} and \ref{p4} are NP-hard. 
Since we can validate $D_I$ in polynomial time for Problem \ref{p4}, it is in NP.
Since Problem \ref{p4} is NP-hard and in NP, it is NP-complete.
\end{proof}

\subsection{Workflow of RatGene}\label{subsec:parent}
Algorithm \ref{rat} describes the workflow of RatGene, which
integrates two constraint-based models $N_c$ and $N_e$ by Algorithm \ref{mi},
determines a modification strategy for growth-coupled production by Algorithm \ref{a1},
and reduces the strategy size by Algorithm \ref{a4}.
Algorithm \ref{a4} calls Functions \ref{f1} and \ref{f2} for the reduction of gene deletions and additions, respectively. 
Algorithms \ref{mi} to \ref{a4} and Functions \ref{f1} and \ref{f2} are described in the following subsections.

\addtocounter{algorithm}{-1}
\begin{algorithm}
\caption{RatGene}\label{rat}
\begin{algorithmic}[1]
    \renewcommand{\algorithmicrequire}{\textbf{Input:}}
    \renewcommand{\algorithmicensure}{\textbf{Output:}}
    \Require constraint-based model $N_c$ and $N_e$, target metabolite
    \Ensure A modification strategy $D_{min}$ for growth-coupled production
    \If{$N_e \notin \varnothing$}
        \State model integration of $N_c$ and $N_e$ to obtain $N_i$ by Algorithm \ref{mi}
    \EndIf
    \State obtain the modification strategy $D$ by Algorithm \ref{a1}
    \State obtain $D_{min}$ by Algorithm \ref{a4}
\end{algorithmic}
\end{algorithm}

\subsection{Model Integration}\label{subsec:integrate}
Problems \ref{p2} (Prob-gDel) and \ref{p4} (Prob-gDel-Add) described in Section \ref{subsec:prodef} are the main problems in this study. Prob-gDel-Add is the same as Prob-gDel after integrating two constraint-based models. 
For the integration of two constraint-based models, a {\bf core model} $N_c = \{ R_c, M_c, G_c, S_c, lb_c, ub_c, h_c \}$ is selected first, and then the other is defined as the {\bf edge model} $N_e = \{ R_e, M_e, G_e, S_e, lb_e, ub_e, h_e \}$. 
Reactions included in the edge model but not in the core model are defined as $R_{ec} = R_e \backslash \{R_e \cap R_c\}$; their lower and upper bounds and GPR rules are represented by $lb_{ec}$, $ub_{ec}$, and $h_{ec}$, respectively. 
Metabolites and genes associated with $R_{ec}$ are represented by $M_{ec}$ and $G_{ec}$, respectively.
Then, additional reactions $R_{ec}$, metabolites $M^{ec} = M_{ec} \backslash \{M_{ec} \cap M_c\}$ and genes $G^{ec} = G_{ec} \backslash \{G_{ec} \cap G_c\}$ are added to the core model to form the integrated model: $N_I = \{ R_I = R_{ec} \cup R_c, M_I = M^{ec} \cup M_c, G_I = G^{ec} \cup G_c, lb_I = lb_{ec} \cup lb_c, ub_I = ub_{ec} \cup ub_c, h_I = h_{ec} \cup h_c \}$. 
The integrated stoichiometric matrix $S_I$ is constructed by a horizontal concatenation of two matrices
as shown in Figure \ref{fig:matrix}. 
The left-part matrix is formed by a vertical concatenation of a zero matrix with $|M^{ec}|$ rows and $|R_c|$ columns below the matrix $S_c$. 
And the right-part matrix is made by a vertical concatenation of a zero matrix and the matrix $S_{ec}$, which corresponds to metabolites $M_{ec}$ and reactions $R_{ec}$ with $|M_{ec}|$ rows and $|R_{ec}|$ columns. 
This results in the integrated model $N_I = \{ R_I, M_I, G_I, S_I, lb_I, ub_I, h_I \}$ retain the same format as the core model. 
Algorithm \ref{mi} is the pseudo-code for the model integration.

\begin{figure}[ht]
    \centering
    \includegraphics[scale=0.5]{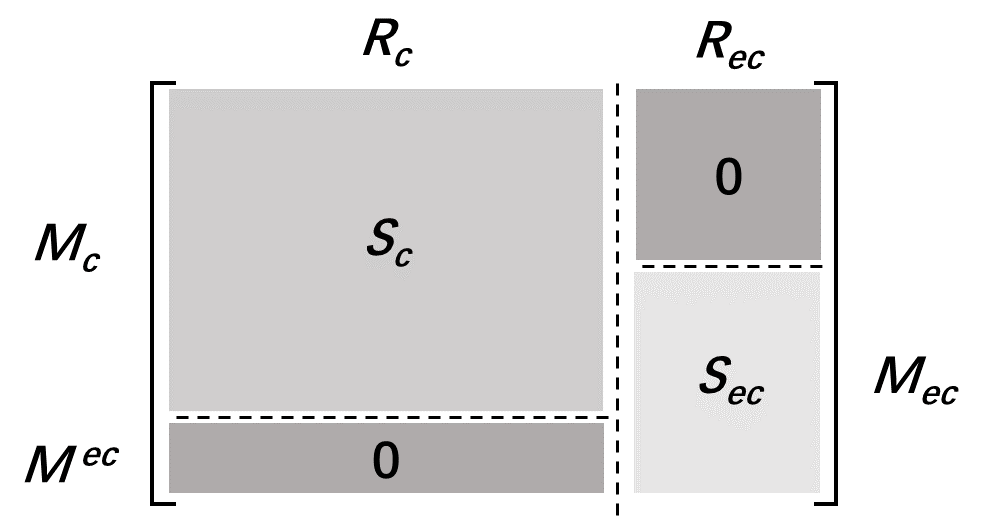}
    \captionsetup{justification=centering}
    \caption{The structure of the integrated stoichiometric matrix $S_I$.}
    \label{fig:matrix}
\end{figure}

\begin{algorithm}
\caption{Model Integration}\label{mi}
\begin{algorithmic}[1]
    \renewcommand{\algorithmicrequire}{\textbf{Input:}}
    \renewcommand{\algorithmicensure}{\textbf{Output:}}
    \Require core model $N_c$, edge model $N_e$
    \Ensure integrated model $N_I=\{ R_I, M_I, G_I, S_I, lb_I, ub_I, h_I \}$
    \State $R_{ec} \gets R_e \backslash \{R_e \cap R_c\}$
    \State obtain $M_{ec}$, $G_{ec}$, $h_{ec}$, $lb_{ec}$, and $ub_{ec}$ corresponding to $R_{ec}$
    \State $R_I = R_{ec} \cup R_c, lb_I = lb_c \cup lb_{ec}, ub_I = ub_c \cup ub_{ec}, h_I = h_c \cup h_{ec}$
    \State $M^{ec} \gets M_{ec} \backslash \{M_{ec} \cap M_c\}$, $G^{ec} \gets G_{ec} \backslash \{G_{ec} \cap G_c\}$
    \State $M_I \gets M^{ec} \cup M_c$, $G_I \gets G^{ec} \cup G_c$
    \State $S_{Il} \gets [S_c; zeros(|M^{ec}|,|R_c|)]$~~~//vertical concatenation
    \For {$i = 1$ to $|M_I|$}
        \If {$M_I(i) \in M_{ec}$}
            \State $S_{Ir}[i][:] \gets S_e[position(M_I(i) \ in \ M_e)][:]$
        \Else
            \State $S_{Ir}[i][:] \gets 0$
        \EndIf
    \EndFor
    \State $S_I \gets [S_{Il} | S_{Ir}]$~~~//horizontal concatenation
\end{algorithmic}
\end{algorithm}

For example, Figures \ref{fig:net2}(A) and \ref{fig:net2}(B) represent a core and an edge model, respectively. 
Figure \ref{fig:net2}(C) is the integrated model consisting of the core model and a part of the edge model. 
$R_c$ and $R_{ec}$ are \{R1,R2,R3,R6,R7,R8,R9,R10,R11\} and \{R4,R5\}, respectively. 
According to the definition, $M_c$, $M_{ec}$, and $M^{ec}$ are then \{C1,C2,C3,C5,C6,C7\}, \{C1,C3,C4,C6\}, and \{C4\}, respectively. Thus, $S_c$ is the stoichiometric matrix of the core model and  $S_{ec}$ is a part of the stoichiometric matrix of the edge model. And this $S_{ec}$ is a matrix with a size of four rows and two columns, which is specifically composed with the interaction of rows corresponding to \{C1,C3,C4,C6\} and columns corresponding to \{R4,R5\} from the original stoichiometric matrix $S_e$ in the edge model.

\subsection{Strategy Generation}\label{subsec:strgen}
In the next phase of RatGene, a ratio-based constraint is iteratively applied to systematically construct MILP problems. These problems can then be used to determine a modification strategy.
In each iteration of the loops, a different ratio-based MILP problem is formulated and solved. Each MILP problem ensures the flux balance assumption and sets the upper and lower bounds of the reaction rates.
Due to the need to simulate the functioning of biological systems, certain specific reactions have assigned thresholds for their lower and/or upper bounds.
Typically, thresholds are set to both the minimum rate for cell growth reactions and the maximum rates for oxygen and glucose uptake reactions.
These reactions are crucial for the survival and functioning of microorganisms:
\begin{align}
    &v_{biomass} \geq lb^{min}_{biomass} \\
    &v_{oxygenUptake} \leq ub^{max}_{oxygenUptake} \\
    &v_{glucoseUptake} \leq ub^{max}_{glucoseUptake}
\end{align}
The associations between GPR rules and their corresponding reactions can be represented by linear constraints \cite{tamura2023gene}. 

A predetermined value $\alpha$ restricts the ratio between the target reaction rate $v_{target}$ and the growth reaction rate $v_{biomass}$. 
However, the appropriate value of $\alpha$ is unknown and different for different networks and different target metabolite production reactions. 
Therefore, loops are designed to iteratively assign different constants to find an appropriate value for $\alpha$.

According to the ratio constraint $v_{target} \cdot {v_{biomass}}^{-1}=\alpha$, clearly, $\alpha \in [0 \cdot TMGR^{-1},$ $TMPR \cdot {lb^{min}_{biomass}}^{-1}]$ can be obtained because the values of $v_{target}$ and $v_{biomass}$ are both non-negative and the $lb$ for $v_{biomass}$ should be greater than zero by the problem definition, where $TMPR$ is the theoretical maximum production rate and $TMGR$ is the theoretical maximum growth rate. 
A constant value $TMPR \cdot {lb_{biomass}^{min}}^{-1} \cdot maxLoop^{-1} \cdot loop$ is added to the value of $\alpha$ in each iteration of the loops. 
Here, $maxLoop$ is the total number of iterations and $loop$ is the cumulative number of iterations. 
The objective function of each MILP problem is constructed as minimizing the sum of $l_0$-Norm of the reactions scaled by $TMGR$ and the negative value of the cell growth reaction rate.
Then, the MILP problem is formulated as (\ref{sg}). 
Figure \ref{fig:net2} shows an example of the importance of the extra ratio constraint and the reason why such an objective function is designed.
\begin{align}\label{sg}
    &min \quad -v_{biomass} + TMGR \cdot \|v_Q\|_0 \\
    &s.t. \nonumber \\
    &\quad S \cdot v=0 \nonumber \\
    &\quad p_i \cdot lb_i \leq v_i \leq p_i \cdot ub_i \nonumber \\
    &\quad p_i=h_i(g) \nonumber \\
    &\quad v_{biomass} \geq lb^{min}_{biomass} \nonumber \\
    &\quad v_{oxygenUptake} \leq ub^{max}_{oxygenUptake} \nonumber \\
    &\quad v_{glucoseUptake} \leq ub^{max}_{glucoseUptake} \nonumber \\
    &\quad \frac{v_{target}}{v_{biomass}}=\alpha \nonumber \\
    &\quad 0 \leq \alpha \leq \frac{TMPR}{lb^{min}_{biomass}} \nonumber \\
    &\quad Q=\{q \ | \ \exists h_q\} \nonumber
\end{align}
\\

For example, for a core and an edge model as shown in
Figures \ref{fig:net2}(A) and \ref{fig:net2}(B),
the integrated model is obtained as shown in Figure \ref{fig:net2}(C).
The edge model shares the same internal reaction R6 with the core model. 
The input is reaction R1 which plays a role of the nutrient uptake reaction, and the outputs are reactions R10 and R11, where R10 is the cell growth reaction and R11 is the target metabolite production reaction. 

In Figure \ref{fig:net2}(C), $\{R1, R2, R3, R10\}$, $\{R1, R7, R8, R11\}$, $\{R1, R6, R7, R8, R9\}$, $\{R1,$
$R4, R5, R10, R11\}$, and $\{R1, R4, R5, R6, R9, R10\}$ are the five elementary modes whose linear combinations satisfy the constraint \ref{cs}. 
Among these elementary modes, $\{R1, R4, R5, $
$R10, R11\}$ is included in neither the core model nor the edge model, so it is the newly generated elementary mode as a result of the integration of two models. 
To achieve the growth-coupled production, $\{R1, R4, R5, R10, R11\}$ can be the choice under the condition that the ratio between R10 and R11 is determined to be 1 and this could even satisfy the criterion for the worst-case analysis. 
However, a linear combination of $\{R1, R2, R3, R10\}$ and $\{R1, R7, R8, R11\}$ is also possible to make the ratio fixed at 1 in Figure \ref{fig:net2}(C). 
The number of reactions required by the linear combination of  $\{R1, R2, R3, R10\}$ and $\{R1, R8, R9, R11\}$ is seven which is greater than the number of reactions for $\{R1, R4, R5, R10, R11\}$. Therefore, minimizing $l_0$-Norm of the reactions and imposing $\alpha=1$ results in finding an effective reaction deletion strategy for this example.  
When multiple solutions exist in minimizing $l_0$-Norm,
GR is maximized.
To this end, $TMGR$ is multiplied by the $l_0$-Norm to reflect the priority in the objective function.
The feasible solution is deleting reactions $\{R2, R3, R6, R7, R8, R9\}$ in the integrated model, Figure \ref{fig:net2} (C), that is, deleting $\{R2, R3, R6, R7, R8, R9\}$ and adding $\{R4, R5\}$ to the core model, Figure \ref{fig:net2} (A). 

To obtain the modification strategies at the gene level, only reactions with GPR rules are considered for the $l_0$-Norms. 
The modification strategy can be represented by deleting $\{g1, g2, g3, g4, g6, g7\}$ in the integrated model, that is, deleting $\{g1, g2, g3, g4, g6, g7\}$ and adding $\{g8, g9\}$ in the core model. 
It is to be noted that the modification strategy at the gene level does not always exist for a designated reaction deletion strategy.
In RatGene, the modification strategy at the gene level is directly determined by MILP.

The obtained modification strategy consists of six gene deletions and two gene additions for the core model. 
The reduced modification strategy by the process described in Section \ref{subsec:sizred} consists of five gene deletions $\{g1,g2,g4,g6,g7\}$ and two gene additions $\{g8,g9\}$.
Algorithm \ref{a1} summarizes the process of strategy generation.

\begin{figure}
    \centering
    \includegraphics[scale=0.6]{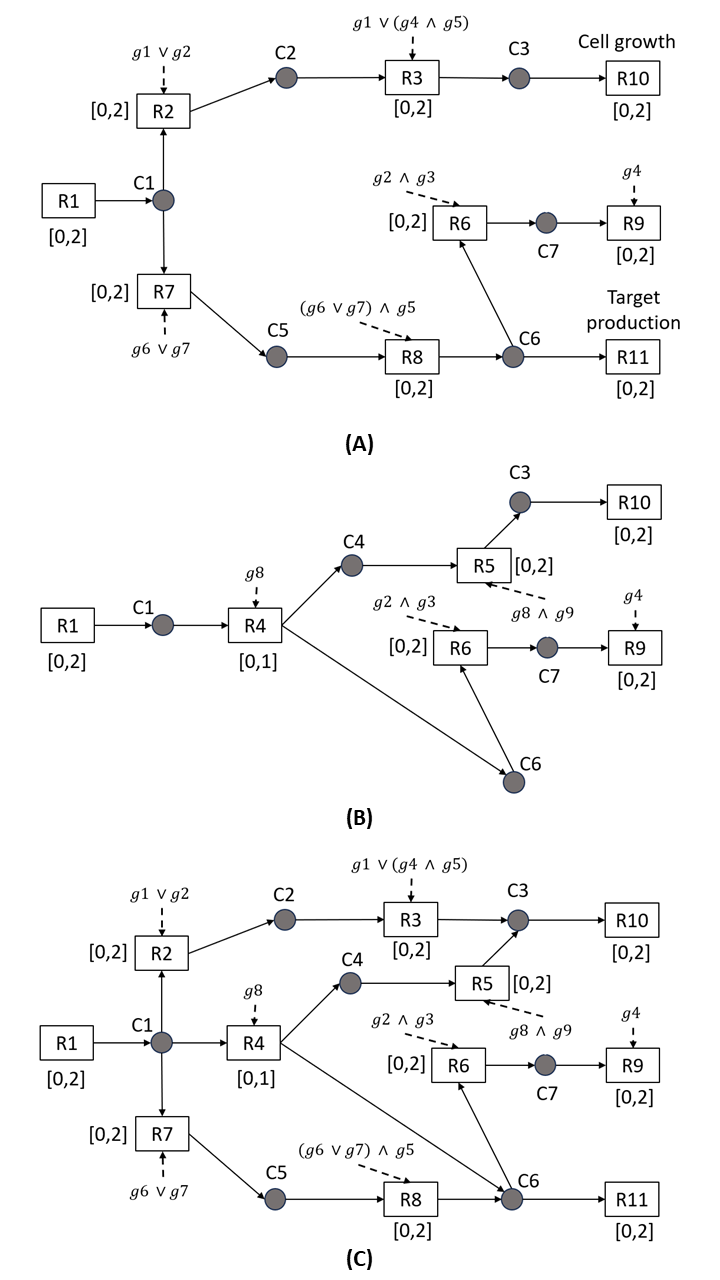}
    \captionsetup{justification=centering}
    \caption{A Toy Constraint-based model. R1 to R11, C1 to C7, and $g1$ to $g9$ are reactions, metabolites, and genes, respectively. Bounds and GPR rules are shown as well. (A) The core model. (B) The edge model. (C) The integrated model of the core model and the edge model.}
    \label{fig:net2}
\end{figure}

\begin{algorithm}
\caption{Strategy Generation}\label{a1}
\begin{algorithmic}[1]
    \renewcommand{\algorithmicrequire}{\textbf{Input:}}
    \renewcommand{\algorithmicensure}{\textbf{Output:}}
    \Require constraint-based model $N$, $v_{target}$, $maxLoop$, $TMGR$
    \Ensure $D$
    \State Initialization: $\alpha \gets 0$, $loop \gets 1$
    \State Form problem $P_{V}^{v_{target}}$, $V$ derives from $(1)$ to $(5)$
    \State $TMPR \gets lp(P_{V}^{v_{target}})$
    \State $gap \gets \frac{TMPR}{lb_{biomass}^{min} \cdot maxLoop}$
    \If {$TMPR > 10^{-3}$}
        \State Form problem $P_{V^{(\ref{sg})}}^{f^{(\ref{sg})}}$ based on $(\ref{sg})$
        \While {$loop \leq maxLoop$}
            \State $\alpha \gets \alpha + gap$ //Update $\alpha$ to $(\ref{sg})$
            \State $pool \gets milp(P_{V^{(\ref{sg})}}^{f^{(\ref{sg})}})$
            \State $D \gets verify(pool)$
            \If {$D \neq \varnothing$}
                \State {\bf return}
            \EndIf
            \State $loop \gets loop + 1$
        \EndWhile
    \EndIf   
\end{algorithmic}
\end{algorithm}

\subsection{Strategy Size Reduction}\label{subsec:sizred}

Algorithm \ref{a4} reduces the sizes of the modification strategy $D$ consisting of deletion $D_d$ and addition $D_a$. 
Algorithm \ref{a4} employs Functions \ref{f1} and \ref{f2} as submodules and its pseudo-code is described below.

\floatname{algorithm}{Algorithm}
\begin{algorithm}[ht]
\caption{Size Reduction for a Modification Strategy}\label{a4}
\begin{algorithmic}[1]
    \renewcommand{\algorithmicrequire}{\textbf{Input:}}
    \renewcommand{\algorithmicensure}{\textbf{Output:}}
    \Require the modification strategy $D$, $h(g)$
    \Ensure $D_{min}$
    \State $D=D_a \cup D_d$ //divide $D$ into addition and deletion
    \For {$i=1$ to $n$} //deletion strategy
        \If {$h_i(D_d)==0$}
            \State {\bf D}$^{i}=MinDel(h_i(g),D_d)$ //Function \ref{f1}
        \EndIf
    \EndFor
    \For {$j=1$ to $n$} //addition strategy
        \If {$h_j(D_a)==1$}
            \State {\bf A}$^{j}=MinAdd(h_j(g),D_a)$ //Function \ref{f2}
        \EndIf
    \EndFor
    \State $D_{min}=min \sum g$
    \State \qquad s.t. \; $r_i=\bigvee (${\bf conjunction} ({\bf D}$^i)), i=1,...n_d$
    \State \qquad \qquad $r_j=\bigvee (${\bf conjunction} ({\bf A}$^j)), j=1,...n_a$
    \State \qquad \qquad $g=\{0,1\}$~~//~1 indicates deletion in {\bf D} and addition in {\bf A}.
    \State \qquad \qquad $r_i,r_j=1$
\end{algorithmic}
\end{algorithm}

\floatname{algorithm}{Function}
\addtocounter{algorithm}{-3}
\begin{algorithm}[ht]
\caption{Size Reduction for Gene Deletions}\label{f1}
\begin{algorithmic}[1]
\Function{MinDel}{$h_{i}(g)$, $D_d$}
    \State {\bf D}$~= D_1 \gets \varnothing$
    \If {$h_i(g)=\bigvee_{t=1}^{T} C_t$}
        \For {$C_1$ to $C_T$}
            \If {clause $C_t$ has a single literal $g_k$ and $g_k \in D_d$}
            \State $D_{z} = D_{z} \cup g_k$ for all $z$
            \ElsIf {clause $C_t$ has multiple literals}
            \State {\bf D} $=$~~{\bf D}~$ \times MinDel(C_t,D_d)$~~/*$\times$ means direct product.
            \EndIf
        \EndFor
    \ElsIf {$h_i(g)=\bigwedge_{t=1}^{T} C_t$}
        \State $len \gets min\{\|MinDel(C_1,D_d)\|,...,\|MinDel(C_T,D_d)\|\}$~~/*$len > 0$
        \For {$C_1$ to $C_T$}
            \If {$\|MinDel(C_t,D_d)\|==len$}
                \State {\bf D}$=$ {\bf D} $\times MinDel(C_t,D_d)$
            \EndIf
        \EndFor
    \ElsIf {$h_i(g)=g_k$ and $g_k \in D_d$}
        \State {\bf D}~$= g_k$
    \EndIf
    \State \Return {\bf D}
\EndFunction
\end{algorithmic}
\end{algorithm}

Function \ref{f1} reduces the size of deletions $D_d$ by considering each reaction forced to be 0 by GPR rules in the core model.
The GPR rule $h_i(g)$ associated to the reaction $v_i$ in a constraint-based model can be categorized into three groups based on its formula: 
\begin{align*}
    (1) h_i(g)=\bigwedge_{t=1}^{T} C_t \\
    (2) h_i(g)=\bigvee_{t=1}^{T} C_t \\
    (3) h_i(g)=g_k.
\end{align*}
$C_t$ denotes a clause, which can be classified into one of the three types described above.
$T$ is the number of clauses. 
Suppose that $h_i(g)$ is Type (1).
For each clause $C_t$ for $v_i$, if the clause is a gene $g_i$, that is Type (3), then $g_k$ is added to the deletion list $D_z$.
If a clause is either Type (1) or (2), Function \ref{f1} is recursively called.
Since Function \ref{f1} may return multiple $D_z$, denoted by {\bf D},
the deletion lists are updated by the direct product operation.
Next, suppose that $h_i(g)$ is Type (2).
Function \ref{f1} is recursively called for each clause and $len$ represents the minimum size of deletions in $MinDel(C_t, D_d)$.
It is to be noted that multiple clauses have the size $len$.
For such clauses, Function \ref{f1} is recursively called and the deletion list {\bf D} is updated by the direct product of {\bf D} and the result of Function \ref{f1}.
Finally, suppose that $h_i(g)$ consists of a single gene $g_k$, that is Type (3).
Then, $g_k$ is added to {\bf D}.
Function \ref{f1} returns {\bf D} as the deletion list with the reduced size.

Similarly, Function \ref{f2} reduces the size of additions $D_a$ by considering reactions forced to be 1 by GPR rules outside the core model.
Suppose that $h_i(g)$ is Type (1).
If a clause $C_t$ consists of a single gene $g_k$, that is Type (3), then $g_k$ is added to the addition list $A_z$.
If a clause is either Type (1) or (2), Function \ref{f2} is recursively called.
Since Function \ref{f2} may return multiple $A_z$, denoted by {\bf A},
the addition lists are updated by the direct product operation.
Next, suppose that $h_i(g)$ is Type (2).
Function \ref{f2} is recursively called for each clause and $len$ represents the minimum size of additions in $MinAdd(C_t, D_a)$.
It is to be noted that multiple clauses have the size $len$.
For such clauses, Function \ref{f2} is recursively called and the addition list {\bf A} is updated by the direct product of {\bf A} and the result of Function \ref{f2}.
Finally, suppose that $h_i(g)$ consists of a single gene $g_k$, that is Type (3).
Then, $g_k$ is added to {\bf A}.
Function \ref{f2} returns {\bf A} as the addition list with the reduced size.

\begin{algorithm}[ht]
\caption{Size Reduction for Gene Additions}\label{f2}
\begin{algorithmic}[1]
\Function{MinAdd}{$h_i(g)$, $D_{a}$}
    \State {\bf A}$ = A_1 \gets \varnothing$
    \If {$h_i(g)=\bigwedge_{t=1}^{T} C_t$}
        \For {$C_1$ to $C_T$}
            \If {clause $C_t$ has a single literal $g_k$ and $g_k \in D_a$}
            \State $A_z = A_z \cup g_k$ for all $z$
            \ElsIf {clause $C_t$ has multiple literals}
               \State {\bf A} = {\bf A} $\times MinAdd(C_t,D_d)$
            \EndIf
        \EndFor
    \ElsIf {$h_i(g)=\bigvee_{t=1}^{T} C_t$}
        \State $len \gets min\{\|MinAdd(C_1,D_a)\|,...,\|MinAdd(C_T,D_a)\|\}$~~//~$len > 0$
        \For {$C_1$ to $C_T$}
            \If {$\|MinAdd(C_t,D_a)\|==len$}
                \State {\bf A} = {\bf A} $\times MinAdd(C_t,D_d)$
            \EndIf
        \EndFor
    \ElsIf {$h_i(g)=g_k$ and $g_k \in D_a$}
        \State {\bf A} $=g_k$
    \EndIf
    \State \Return {\bf A}
\EndFunction
\end{algorithmic}
\end{algorithm}

Finally, Algorighm \ref{a4} conducts integer linear programming (ILP) to obtain the modification strategy with the minimum size using the constraints obtained by Functions \ref{f1} and \ref{f2}.

\section{Computational Experiments}\label{sec:comexp}
All the procedures executed in the computational experiments in this study were implemented on a Ubuntu 20.04 machine with an AMD Ryzen Threadripper3 3970X CPU of 3.70GHz 32C/64T. The running environment is based on IBM ILOG CPLEX 12.10, COBRA Toolbox 2022 and MATLAB R2019b.
An auxiliary exchange reaction was temporarily added to the model to simulate target metabolite production if the target metabolite did not have a production reaction.

The computational experiments focused on three key aspects: the success rate, the modification size, and the computation time.
All the results were filtered by worst-case analysis, where only the minimum PR was evaluated under the GR maximization for each modification strategy.
We used four different datasets from two major model species, \textit{S. cerevisiae} and \textit{E. coli}, downloaded from the BiGG database  \cite{norsigian2020bigg}.
We constructed six datasets using these four datasets for the computational experiments for Prob-gDel and Prob-gDel-Add. 

Table \ref{t1} represents the details of these six datasets. 
The iMM904 is a constraint-based model for \textit{S. cerevisiae}, 
while both iJR904 and iML1515 are constraint-based models for \textit{E. coli}. 
These three datasets were input for Prob-gDel, aiming to determine the deletion strategy. 
For Prob-gDel-Add, the edge model \textit{E. coli} iJR904 was incorporated into the core model iMM904, while the edge model iND750 of \textit{S. cerevisiae} was combined with the core models iJR904 and iML1515. The resulting integrated datasets exhibited an increase ranging from 300 to 400 in metabolites, 600 to 800 in reactions, and 500 to 600 in genes compared to their original counterparts. 
The efficacy of RatGene was benchmarked against gDel\_minRN and GDLS using these six datasets. 
The allocated computational time for each target metabolite was limited to a specific threshold.
All fluxes with rates $10^{-3}$ or less were treated as having rates equivalent to zero.

\begin{table}[ht]
\captionsetup{justification=centering}
\caption{Datasets}\label{t1}
\centering
\begin{tabular}{ccccc}
\hline
ID & Dataset & Metabolites & Reactions & Genes \\
\hline
1 & iMM904 & 1226 & 1577 & 905 \\
2 & iMM904+iJR904 & 1518 & 2186 & 1506 \\
3 & iJR904 & 761 & 1075 & 904 \\
4 & iJR904+iND750 & 1139 & 1892 & 1419  \\
5 & iML1515 & 1877 & 2712 & 1516 \\
6 & iML1515+iND750 & 2198 & 3560 & 2037 \\
\hline
\end{tabular}
\end{table}

\subsection{deletion-addition Problem: Prob-gDel-Add}\label{sub-ad}
Table \ref{t2} compares the performance of RatGene, gDel\_minRN, and GDLS for Prob-gDel-Add on the three datasets. 
The row labeled ID 1 represents the number of metabolites whose theoretical maximum PR is more than 0.001. 
The rows labeled ID 2 to 4 represent the number of target metabolites for which each method could determine the modification strategy for growth-coupled production. 
The rows labeled ID 5 to 7 represent the success ratio.
RatGene had the highest number of success and success rates on Datasets 4 and 6, where the success rate exceeded 26\% on Dataset 4 and was nearly 40\% on Dataset 6. 
gDel$\_$minRN had the best performance for the number of successes and success rate on Dataset 2. 
The differences in success rates of RateGene and gDel\_minRN are less than 5\% for Datasets 2 and 4.
The success ratio of RatGene was almost double that of gDel\_minRN for Dataset 6.
The success ratio of GDLS was much less than that of the best method across all datasets. 

The rows labeled IDs 8 and 9 represent the size of the union and the intersection of the successful cases by RatGene and gDel\_minRN.
The rows labeled IDs 10 and 11 represent their percentages.
If we executed both RatGene and gDel\_minRN and adopted better solutions, the success rate exceeded 35\%.
The size of the union is almost two to three times the size of the intersection.

\setlength{\tabcolsep}{6pt}
\begin{table}[ht]
\captionsetup{justification=centering}
\caption{Number of Success and Success Ratio(\%)}\label{t2}
\centering
\begin{tabular}{clccc}
\hline
ID&Dataset & iMM904+iJR904 & iJR904+iND750 & iML1515+iND750 \\
 && Dataset 2 & Dataset 4 & Dataset 6 \\ \hline
1&TMPR$>10^{-3}$ & 825 & 569 & 1295 \\
2&RatGene & 270 & \textbf{150} & \textbf{506} \\
3&gDel$\_$minRN  & \textbf{310} & 127 & 285 \\
4&GDLS & 61 & 79 & 33 \\
5&RatGene(\%) & 32.73\% & \textbf{26.36\%} & \textbf{39.07\%} \\
6&gDel$\_$minRN(\%) & \textbf{37.58\%} & 22.32\% & 22.01\% \\
7&GDLS(\%) & 7.39\% & 13.88\% & 2.55\% \\ \hline
8&$^a$Union & 395 & 204 & 593 \\
9&$^b$Intersection & 185 & 73 & 198 \\
10&Union(\%) & 47.88\% & 35.85\% & 45.79\% \\
11&Intersection(\%) & 22.42\% & 12.83\% & 15.29\% \\
\hline
\multicolumn{4}{l}{\footnotesize{$^a$ The union results of RatGene and gDel$\_$minRN.}} \\
\multicolumn{4}{l}{\footnotesize{$^b$ The intersection results of RatGene and gDel$\_$minRN.}} \\
\end{tabular}
\end{table}

Table \ref{t3} represents the average sizes of additions and deletions in the modification strategies derived by the three methods. 
Rows labeled IDs 1 to 3 represent the size of deletions for each method and dataset, while rows labeled IDs 4 to 6 represent the size of additions.
RatGene yielded smaller deletions compared to gDel\_minRN for all three datasets. 
GDLS returned smaller average deletion sizes across all these datasets, but this result is less competitive when considering its low success rates.
Regarding the addition sizes, RatGene consistently returned the smallest average sizes, all below 70. 
In contrast, the smallest average size for additions by gDel\_minRN was 74.46.
Although GDLS had the smallest average deletion sizes, it had much larger average addition sizes when compared to RatGene and gDel\_minRN.
For both RatGene and gDel\_minRN, the average addition sizes were lower than the average deletion sizes for all datasets. 
In contrast, GDLS displayed the reverse results.

\setlength{\tabcolsep}{5pt}
\begin{table}[ht]
\captionsetup{justification=centering}
\caption{Average Modification Size}\label{t3}
\centering
\begin{tabular}{rcccc}
\hline
ID&Dataset & iMM904+iJR904 & iJR904+iND750 & iML1515+iND750 \\
&& Dataset 2 & Dataset 4 & Dataset 6 \\ \hline
1&RatGene($^a$Del) & 553.19 & 608.91 & 1014.25 \\
2&gDel$\_$minRN(Del) & 684.94 & 676.98 & 1127.15 \\
3&GDLS(Del) & 1.07 & 0.75 & 0 \\
4&RatGene($^a$Add) & 67.13 & 31.24 & 50.89 \\
5&gDel$\_$minRN(Add) & 168.74 & 74.46 & 88.14 \\
6&GDLS(Add) & 600.49 & 514.99 & 521.00 \\
\hline
\multicolumn{4}{l}{\footnotesize{$^a$ The size of deletions strategies.}} \\
\multicolumn{4}{l}{\footnotesize{$^b$ The size of addition strategies.}} \\
\end{tabular}
\end{table}

Table \ref{t4} displays the average computation time for each metabolite of all successful cases.
RatGene had the lowest average computational time on Dataset 2, while GDLS was the fastest on the other two datasets. 
In computing the modification strategy in each dataset, the average computational time of RatGene for the successful cases was less than one-third to one-sixth of gDel\_minRN.
Based on the results of the success rates in Table \ref{t2}, RatGene made fast computations while ensuring a relatively good performance in terms of success rates.
\begin{table}[ht]
\captionsetup{justification=centering}
\caption{Average Computational Time of Successes (seconds)}\label{t4}
\centering
\begin{tabular}{ccccc}
\hline
Dataset& & RatGene & gDel$\_$minRN & GDLS \\
\hline
iMM904+iJR904&(Dataset 2) & 162.83 & 541.06 & 197.66 \\
iJR904+iND750&(Dataset 4)& 180.37 & 547.18 & 33.45  \\
iML1515+iND750&(Dataset 6)& 85.36 & 585.66 & 20.39 \\
\hline
\end{tabular}
\end{table}

\subsection{Deletion Problem: Prob-gDel}

\setlength{\tabcolsep}{10pt}
\begin{table}[ht]
\captionsetup{justification=centering}
\caption{Number of Success and Success Ratio(\%)}\label{t5}
\centering
\begin{tabular}{rcccc}
\hline
ID&Dataset & iMM904 & iJR904 & iML1515 \\
&& Dataset 1 &  Dataset 3 & Dataset 5 \\ \hline
1&TMPR$>10^{-3}$ & 782 & 510 & 1092 \\
2&RatGene & \textbf{210} & 38 & \textbf{262} \\
3&gDel$\_$minRN  & 172 & \textbf{224} & 211 \\
4&GDLS & 58 & 33 & 1 \\
5&RatGene(\%) & \textbf{26.85\%} & 7.45\% & \textbf{23.99\%} \\
6&gDel$\_$minRN(\%) & 21.99\% & \textbf{43.92\%} & 19.32\% \\
7&GDLS(\%) & 7.42\% & 6.47\% & 0.09\% \\ \hline
8&$^a$Union & 296 & 236 & 386 \\
9&$^b$Intersection & 86 & 26 & 87 \\ 
10&Union(\%) & 37.85\% & 46.27\% & 35.35\% \\
11&Intersection(\%) & 11.00\% & 5.10\% & 7.97\% \\
\hline
\multicolumn{4}{l}{\footnotesize{$^a$ The union results of RatGene and gDel$\_$minRN.}} \\
\multicolumn{4}{l}{\footnotesize{$^b$ The intersection results of RatGene and gDel$\_$minRN.}} \\
\end{tabular}
\end{table}

Table \ref{t5} compares the performance of RatGene, gDel\_minRN, and GDLS for Problem 2 on the three datasets. 
Similar to Section \ref{sub-ad},
the row labeled ID 1 represents the number of metabolites whose theoretical maximum PR is more than 0.001. 
The rows labeled ID 2 to 4 represent the number of target metabolites for which each method could determine the modification strategy for growth-coupled production. 
The rows labeled ID 5 to 7 represent the success ratio.
RatGene had the highest number of successes and success rates on Datasets 1 and 5, whereas gDel$\_$minRN had the highest on Dataset 3. 
RatGene and GDLS had nearly identical low success rates on Dataset 3. 
GDLS exhibited a notably lower success rate and fewer successes compared to the other two methods across all datasets.
 
The union of successful results generated by RatGene and gDel\_minRN promised to have the lowest success rate of over 35\%. 
The intersection between RatGene and gDel\_minRN only contributed a small amount, indicating that RatGene was efficient in obtaining successful deletion strategies that gDel\_minRN could not find.

Table \ref{t6} represents the average sizes of deletion strategies derived by the three methods. RatGene produced successful deletion strategies with smaller average sizes than gDel\_minRN for all three datasets. In contrast, the average strategy size of GDLS was fewer than ten across all these datasets. For both RatGene and gDel\_minRN, the average deletion sizes for iMM904, iJR904, and iML1515 were lower than the results for the associated integrated datasets in Table \ref{t3}. 
However, the average sizes for deletion strategies generated by GDLS on iMM904, iJR904, and iML1515 were either equal to or larger than those for their respective integrated datasets in Table \ref{t3}.

\setlength{\tabcolsep}{6pt}
\begin{table}[ht]
\captionsetup{justification=centering}
\caption{Average Deletion Size}\label{t6}
\centering
\begin{tabular}{cccc}
\hline
Dataset & iMM904 & iJR904 & iML1515 \\
& Dataset 1 & Dataset 3 & Dataset 5 \\ \hline
RatGene & 480.79 & 409.53 & 865.72 \\
gDel$\_$minRN & 621.88 & 523.31 & 967.33 \\
GDLS & 6.53 & 2.52 & 0 \\
\hline
\end{tabular}
\end{table}

Table \ref{t7} represents the average computation time for each metabolite of all the successes on Datasets 1, 3, and 5. 
GDLS had the lowest average computational time on the datasets, while the success ratio of GDLS is very low. 
The RatGene took less than 50\% of the time to compute the deletion strategy on all three datasets compared to gDel$\_$minRN. 
It was observed that RatGene computed much faster for the deletion-addition strategies in Table \ref{t4} compared to the deletion strategies, while gDel$\_$minRN took longer to compute on the integrated datasets than on its original datasets.
\begin{table}[ht]
\captionsetup{justification=centering}
\caption{Average Computational Time of Successes (seconds)}\label{t7}
\centering
\begin{tabular}{ccccc}
\hline
Dataset& & RatGene & gDel$\_$minRN & GDLS \\
\hline
iMM904&(Dataset 1) & 248.39 & 532.24 & 134.46 \\
iJR904&(Dataset 3)& 241.77 & 520.68 & 43.22 \\
iML1515&(Dataset 5)& 259.26 & 570.28 & 3.72 \\
\hline
\end{tabular}
\end{table}

\section{Discussion and Conclusion}
\subsection{Discussion}\label{subsec:discuss}
Since the performance of RatGene and gDel\_minRN was much better than GDLS,
we focus on RatGene and gDel\_minRN in the following discussion.
RatGene solves Prob-gDel-Add by (1) integrating two constraint-based models and (2) determining gene deletion strategies for the integrated model.
For solving Prob-gDel-Add, (2) of RatGene can be replaced with other algorithms such as gDel\_minRN and GDLS because their purposes are the same.
Tables \ref{t2} to \ref{t4} compare the performance of such three algorithms for Prob-gDel-Add.
Similarly, Prob-gDel can be solved by gDel\_minRN and GDLS, but (2) of RatGene also can solve it.
Tables \ref{t5} to \ref{t7} compare the performance of the three algorithms for solving Prob-gDel.

Tables \ref{t2} and \ref{t5} show that RatGene and gDel\_minRN exhibited strong performance for most datasets regarding both Prob-gDel-Add and Prob-gDel. 
While RatGene outperforms gDel\_minRN for certain datasets, gDel\_minRN is superior for others.
Since the union of their successful cases was large and the intersection was small, RatGene and gDel\_minRN can complement each other effectively.
Therefore, by conducting both RatGene and gDel\_minRN simultaneously, the performance can be significantly improved compared to using each individually.

Tables \ref{t3} and \ref{t6} display that RatGene derived smaller modification strategies than gDel\_minRN. 
Table \ref{t3} indicates that the sizes of both deletions and additions of gDel\_minRN were larger than those of RatGene.
These differences came from the differences in the behavior of gDel\_minRN and (2) of RatGene.
Consequently, both RatGene and gDel\_minRN can complementarily identify effective modification strategies for various target metabolites.
The deletion sized by RatGene and gDel\_minRN may further be reduced by the trial-and-error-based methods proposed in \cite{tamura2022trimming, tamura2023metnetcomp}.

Tables \ref{t4} and \ref{t7} indicate that the computational speed of RatGene was more than three times faster than gDel\_miRN in Prob-gDel-Add, and more than twice as fast in Prob-gDel.
Therefore, running both RatGene and gDel\_minRN simultaneously improved the success rate, but it also increased the computational time by more than threefold.

\begin{table}[t]
\captionsetup{justification=centering}
\caption{Increase in the number of success cases in the integrated models when compared to the original models}\label{t8}
\centering
\begin{tabular}{lccc}
\hline
Dataset & iMM904+iJR904 & iJR904+iND750 & iML1515+iND750 \\
& Dataset 2 &Dataset 4 &Dataset 6 \\ \hline
RatGene & 60 & \textbf{112} & \textbf{244} \\
gDel$\_$minRN & \textbf{138} & -97 & 74 \\
GDLS & 3 & 46 & 32 \\
\hline
TMPR$>10^{-3}$ & 43 & 59 & 203 \\
\hline
\end{tabular}
\end{table}

It is easier to find a growth-coupled production strategy in the integrated model created by $N_1$ and $N_2$ than in the model created by $N_1$ alone.
Table \ref{t8} shows the increase in the number of growth-coupled production strategies identified by each method after the model integration when compared to the original models. 
The last row of Table \ref{t8} shows the number of target metabolites whose TMPRs were greater than zero in the integrated model but zero in the original model. 
On all data sets, the increase of target metabolites for which modification strategies were found by RatGene was larger than the increase of metabolites whose TMPRs were greater than zero. 
This means that RatGene successfully identified efficient modification strategies for many target metabolites for which strategies for growth-coupled production could not be obtained in the original dataset.
Either RatGene or gDel\_minRN exhibited an increase of more than 40 in the number of successful cases compared to the increase in the number of target metabolites with a TMPR 
$> 10^{-3}$ for each dataset.
For Dataset 4, the performance of gDel\_minRN was not as good as in the original Dataset 3. This was due to the time limit set for each metabolite and the integrated models are larger than their original models.
This indicates that RatGene is more efficient for large models than gDel\_minRN.

\begin{table}[ht]
\captionsetup{justification=centering}
\caption{Average computational time by RatGene for all target metabolites with TMPR$>10^{-3}$ (seconds).}\label{t9}
\centering
\begin{tabular}{cccc}
\hline
 & iMM904 & iJR904 & iML1515 \\
\hline
Prob-gDel & 442.36 & \textbf{508.89} & 450.95 \\
Prob-gDel-Add & 223.31 & 245.04 & 136.53 \\
\hline
\end{tabular}
\end{table}

Table \ref{t9} shows the average computational time taken by RatGene for all metabolites with  TMPRs greater than zero.
It shows that the average time for Prob-gDel was approximately 13\% longer on dataset iJR904 compared to the other two datasets.
The average time cost for Prob-gDel was more than double the average time cost for Prob-gDel-Add.
These trends can be attributed to the reason that the computation continued until reaching the upper limit of calculation time when a solution was not found.
RatGene primarily functions by iteratively assigning a fixed value to $\alpha$ to determine its value under growth-coupled production conditions. 
This process increases $\alpha$ from 0 to its maximum within a predefined computation time limit. 
Consequently, the computation might be halted before identifying $\alpha$ if its appropriate value is relatively large. 
For deletion problems in the iJR904 model, the success rate of RatGene was initially 7.45\% as shown in Table \ref{t5}.
This rate significantly increased to 26.36\% after integrating the iND750 model, as shown in Table \ref{t2}. 
The integration altered the network topology and the appropriate value of $\alpha$, leading to a reduced average computation time of 245.04 seconds and an enhanced success rate as shown in Tables \ref{t9} and \ref{t2}.

Figure \ref{fig:fig2} compares the average sizes of addition strategies before and after the size reduction. 
The average sizes of the strategies decreased to approximately 50\%--65\% of their original sizes. 
This demonstrates the effectiveness of the proposed process in compressing the size of the modification strategies.

\begin{figure}[ht]
    \centering
    \includegraphics[scale=0.65]{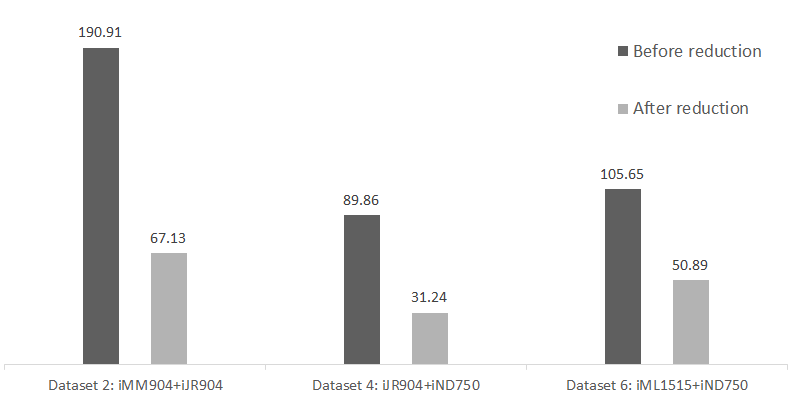}
    \captionsetup{justification=centering}
    \caption{The average size of addition strategies before and after the size reduction.}
    \label{fig:fig2}
\end{figure}

\subsection{Conclusion}\label{subsec:conclu}

In this study, we (1) mathematically defined Prob-gDel and Prob-gDel-Add, (2) proved their NP-completeness, and (3) developed RatGene.
RatGene aims to solve Prob-gDel-Add by integrating two constraint-based models and identifying gene deletion strategies for growth-coupled production. 
RatGene can be used to solve Prob-gDel as well.
The results of the computational experiments showed that RatGene is effective for solving Prob-gDel and Prob-gDel-Add, in particular for target metabolites for which the best existing methods cannot solve.
By using gDel\_minRN and RatGene in a complementary manner, 
the success ratios for Prob-gDel and Prob-gDel-Add can significantly be improved.
RatGene identifies the strategies by growth-to-production ratio-based flux assignments.
This improves the success ratio for identifying the strategies for growth-coupled production while it increases the computation time.

\subsection{Related works}\label{subsec:review}
Numerous studies have been conducted on metabolic design for growth-coupled production.
However, they mainly focus on reaction-level modifications, and did not appropriately consider GPR rules for genome-scale models.
Existing methods can be roughly classified into elementary modes (EMs) and constraint programming-based methods.

The EM-based methods are briefly summarized below.
The CASOP computational framework offers a strategy for strain optimization, using significance measurements of reactions derived from weighted EMs \cite{hadicke2010casop}.
Approach based on Minimal cut sets (MCSs) combining a duality framework in metabolic networks was provided, where the task of counting the MCSs in the original network is simplified to determining the EMs in a dual network \cite{ballerstein2011minimal}.
Constrained MCSs (cMCSs) -based method was proposed to systematically list all comparable gene deletion combinations, supporting the development of effective knockout strategies for coupled product and biomass production \cite{hadicke2011computing}.
Effective preprocessing procedures have been described to be utilized by algorithms that computes cMCSs from EMs, and it was shown how computing MCSs from EMs could be made comparable by a modest adjustment to the integer program \cite{jungreuthmayer2013comparison}.
MCSEnumerator merged two methods to offer the MCS enumerator, which was a novel technique for efficiently enumerating the smallest MCSs by determining shortest EMs in genome-scale metabolic network models with the fewest interventions \cite{von2014enumeration}. 
An approach that focused on how to meaningfully assess and order a set of computed metabolic engineering strategies for growth-coupled product synthesis has also been proposed \cite{schneider2019characterizing}. 
Another work categorized different growth-coupled productions into four groups and extended the MCS framework to calculate the strain designs using the implicit optimality constraints \cite{schneider2021systematizing}.
\\
Constraint programming-based methods are briefly summarized below.
A bilevel programming framework called Optknock was made available for locating knockout strategies \cite{burgard2003optknock}. Various methods based on this bilevel optimization have been proposed subsequently \cite{patil2005evolutionary,pharkya2004optstrain,ranganathan2010optforce}. FastPros offers a unique screening technique in which the potential of a given reaction knockout for the synthesis of a particular metabolite is assessed by the shadow pricing of the constraint in the flux balance analysis \cite{ohno2014fastpros}. GridProd is a grid-based optimization method framework that separates the entire constraint space of the problem into grids and searches for potential deletion strategies in the subspaces \cite{tamura2018grid}. minL1-FMDL, a fast metabolic design listing algorithm, was developed to provide narrowed reaction deletion strategies in polynomial time \cite{tamura2021l1}. Based on this framework, several methods have been proposed to explore the design of metabolic networks under different environmental conditions \cite{tamura2021efficient,ma2021dynamic}.
 \\

\section*{Data Availability}
The developed scripts are available on \href{https://github.com/Ma-Yier/RatGene}{https://github.com/Ma-Yier/RatGene}.

\bibliographystyle{abbrv}
\bibliography{document}


\end{document}